\documentclass[reqno]{amsart}

\usepackage{enumitem}
\usepackage{amsmath,cancel}
\usepackage{amsthm,amsfonts,amssymb,txfonts,pxfonts,empheq}
\usepackage{hyperref}
\usepackage{color}

\newtheorem{proposition}{Proposition}[section]

\newtheorem{corollary}[proposition]{Corollary}
\newtheorem{observation}[proposition]{Observation}
\newtheorem{theorem}[proposition]{Theorem}
\newtheorem{remark}[proposition]{Remark}

\newtheorem {RHP}{Riemann-Hilbert problem}
\newcommand{\lp}{\left(}

\newcommand{\rp}{\right)}
\newcommand{\ls}{\left[}
\def\C{\mathbb C}
\newcommand{\rs}{\right]}
\newcommand{\lb}{\left\{}

\newcommand{\be}{\begin{eqnarray}}
\newcommand{\ee}{\end{eqnarray}}
\newcommand{\tAbb}{\tilde{\mathbb{A}}}
\newcommand{\tBbb}{\tilde{\mathbb{B}}}
\def \le{\left}
\def \1{\mathbf 1}
\def \ri{\right}
\newcommand{\rb}{\right\}}
\newcommand{\la}{\left|}
\newcommand{\ra}{\right|}
\newcommand{\bPsi}{\boldsymbol{\Psi}}
\newcommand{\bPhi}{\boldsymbol{\Phi }}
\newcommand{\tf}{\theta\left(z;x,t\right)}
\newcommand{\sdue}{\sigma_2}
\newcommand{\stre}{\sigma_3}

\newcommand{\hatalpha}{\hat{\alpha}}
\newcommand{\Imag}{\mbox{Im}}
\newcommand{\Real}{\mbox{Re}}

\newcommand*{\mailto}[1]{\href{mailto:#1}{\nolinkurl{#1}}}

\begin{document}

\definecolor{aqua}{rgb}{0.0, 1.0, 1.0}
\baselineskip 15pt plus 1pt minus 1pt

\vspace{0.2cm}
\begin{center}
\begin{Large}
\fontfamily{cmss}
\fontsize{17pt}{27pt}
\selectfont
\textbf{
A degeneration of two-phase solutions of focusing NLS via Riemann-Hilbert problems}
\end{Large}\\
\bigskip
\begin{large} { Marco Bertola $^{\ddagger,\sharp}$, \footnote{Work supported in part by the Natural
    Sciences and Engineering Research Council of Canada
(NSERC).} \footnote{Marco.Bertola@concordia.ca} , Pietro Giavedoni$^\star$\footnote{
Research partially supported by the Austrian Science Fund (FWF) under Grant No.\ Y330 and by FP7 RUSE'S grant RIMMP "Random and Integrable Models in Mathematical Physics"
}
\footnote{addenaro@gmail.com, pietro.giavedoni@univie.ac.at} }
\end{large}
\\
\bigskip
{\small
$^{\ddagger}$ { Department of Mathematics and
Statistics, Concordia University\\ 1455 de Maisonneuve W., Montr\'eal, Qu\'ebec,
Canada H3G 1M8} \\
$^{\sharp}$ { Centre de recherches math\'ematiques, Universit\'e\ de
Montr\'eal } \\
$^\star$ {
Faculty of Mathematics,  University of Vienna\\
Oskar-Morgenstern-Platz 1, 1090 Wien,  Austria\\}
}
\bigskip
{\bf Abstract}
\end{center} 
%
%
%

%

Two-phase solutions of focusing NLS equation are classically constructed out of an appropriate Riemann surface of genus two, and expressed in terms of the corresponding theta-function. We show here that in a certain limiting regime such solutions reduce to some elementary ones called "Solitons on unstable condensate". This degeneration turns out to be conveniently studied by means of basic tools from the theory of Riemann-Hilbert problems. In particular no acquaintance with Riemann surfaces and theta-function is required for such analysis.
\\[1pt]
\hrule
\vskip 12pt

\section{Introduction}

The focusing nonlinear Schr\"odinger equation (abbreviated fNLS)
\begin{align}\label{fNLS}
i q_t + q_{xx} + 2\left| q \right|^2 q = 0, \quad\quad q = q\lp x,t \rp \in \mathbb{C};\,\,  x,t \in\mathbb{R}
\end{align}
plays an important role in modelling phenomena from several fields of physics: Nonlinear Optics, Water Waves and theory of turbulence in plasmas, to quote just some of the main ones (\cite{Chiao},\cite{Zakhar}, \cite{ZakCol}). From the analytic point of view, fNLS exhibits the remarkable feature to be integrable \cite{ZakSha}. As a consequence, many of its solutions can be found via an appropriate Riemann-Hilbert problem. In this paper we wish to consider the following, particular one:\\
Let $E_1$ and $E_3$ be complex numbers with positive imaginary parts such that 
\begin{align}
\Real\lp E_1 \rp > \Real \lp E_3 \rp.
\end{align}
Let \begin{align}\label{movimento}
E_2 = E_3 + \epsilon, \quad\quad \epsilon >0
\end{align}
where $\epsilon$ will be considered to be small. Put
\begin{align}
E_{-j} = \overline{E_j}, \quad \quad j= 1,2,3.
\end{align}
Here the overline in the right-hand side denotes the complex conjugation. Let us consider the oriented contour (see fig. \ref{fig_ContornoGamma}) 
\begin{align}
\Gamma := \lp E_{-3},E_{-2} \rp \cup \lp E_{-2},E_{-1} \rp \cup \lp E_{-1},E_{1} \rp \cup \lp E_{1},E_{2} \rp \cup \lp E_{2},E_{3} \rp.  
\end{align}

\begin{figure}
\begin{center}
\includegraphics[width= 0.45\textwidth]{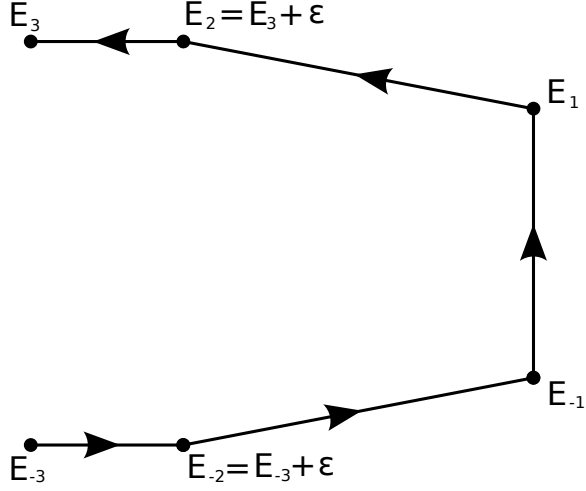}
\caption{\label{fig_ContornoGamma}{\it The closure of the oriented contour $\Gamma$}}
\end{center}
\end{figure}
Following the common usage, for any $z\in \Gamma$  we will denote by $f_\pm(z)$ the (non-tangential) boundary values from the left/right (respectively) of a function $f(z)$ analytic in $\mathbb C \setminus  \overline \Gamma$. By the employment of such notation we implicitly assume that such limits exist.
Finally, fix two arbitrary real numbers $\alpha$ and $\beta$. For every fixed $x$ and $t$, and arbitrarily small $\epsilon$, let us formulate the following Riemann-Hilbert problem (abbreviated RHP).

\begin{RHP}\label{Zafira_RH}
Find a matrix-valued function 
\begin{align}
\bPsi : \mathbb{C}\backslash \overline{\Gamma} \rightarrow Mat\lp 2\times 2, \mathbb{C} \rp
\end{align}
such that:
\begin{description}
\item[$\mathbf{i}$] $\bPsi $ is analytic in $\mathbb{C}\backslash \overline{\Gamma}$.
\item[$\mathbf{ii}$]  For every point  $z\in\Gamma$ its non-tangential limits satisfy
\begin{align}
\bPsi_{+}\lp z \rp = \bPsi_{-}\lp z \rp M\lp z \rp
\end{align}
where 
\begin{align}
M\left(z\right)=\left\{\begin{array}{ll}
                                         e^{-i\left[\tf + \frac{\alpha}{2}\right]\stre}\left(-i\sigma_2\right)e^{i\left[\tf + \frac{\alpha}{2}\right]\stre} & z \in \left(E_2,E_3\right)\\
                                         e^{-i\beta\stre}                                                                                                                                         & z \in \left(E_1,E_2\right)\\
                                         e^{-i\tf\stre}\left(-i\sdue\right)e^{i\tf\stre}                                                                                                 & z \in \left(E_{-1},E_{1}\right)\\
                                         e^{-i\beta\stre}                                                                                                                          & z \in \left(E_{-2},E_{-1}\right)\\     
                            e^{-i\left[\tf+\frac{\alpha}{2}\right]\stre}\left(-i\sdue\right)e^{i\left[ \tf + \frac{\alpha}{2}\right]\stre}     & z \in \lp E_{-3},E_{-2}\rp
\end{array}\right.
\end{align}
Here
\begin{align} \label{def_theta}
\theta\lp z;x,t \rp := 2tz^2  + xz
\end{align}
and 
\begin{align}
\sigma_2 = \ls \begin{array}{cc}
0 & - i \\
 i  & 0
 \end{array}\rs, \quad\quad \sigma_3 = \ls\begin{array}{cc}
1 & 0\\
0 & -1
\end{array}\rs.
\end{align}
\item[\textbf{iii}] The following asymptotic behaviour holds:  
\begin{align}
\bPsi\lp z \rp = \mathbf 1 + \mathcal{O}\lp \frac{1}{z} \rp, \quad\quad \mbox{as} \quad z\rightarrow \infty.
\end{align} 
\item[\textbf{iv}]
The growth condition
\begin{align}
\bPsi\lp z \rp = \mathcal{O}\lp \left| z - E_j \right|^{-\frac{1}{4}} \rp, \quad \quad \mbox{as}\quad z\rightarrow E_{j}
\end{align}
is satisfied for $j=\pm1, \pm 2, \pm 3$.
\end{description}
\end{RHP}
Let us remark that condition $\mathbf{iv}$ is necessary in order to guarantee uniqueness of the solution of  Riemann-Hilbert problem \ref{Zafira_RH}. On the other side, no bounded ones can exist for it. The family of solutions to the parametric RHP \ref{Zafira_RH} will be indicated with $\bPsi\lp z ; \epsilon \rp$, being the dependence on $x$ and $t$ understood. A solution of fNLS can then be produced via the following simple formula.
\begin{align} 
q\lp x,t; \epsilon \rp = -2\lim_{z\rightarrow \infty} z\bPsi_{12}\lp z; \epsilon \rp.    \label{Formula_Beni}
\end{align}

By a variation of the parameters defining Riemann-Hilbert problem \ref{Zafira_RH}, formula (\ref{Formula_Beni}) describes correspondingly the family of the so-called "two-phase" solutions of fNLS (see below). 

{
The aim of this paper is to provide a simple and efficient method to investigate the behaviour of the two-phase solutions (\ref{Formula_Beni}) as $\epsilon$ tends to zero.
The RHP \ref{Zafira_RH} falls within  the  general class of ``quasi-permutation monodromy'' problems \cite{KorSol} whose solution is in principle well known in terms of the Szeg\"o\ kernel of a Riemann surface branched above the endpoints: in our situation the Riemann surface is a hyperelliptic surface (see \ref{Christa}).  Although the solution could be written explicitly, our goal is a different one. Specifically we want to study the behaviour of the solution $q(x,t;\epsilon)$ in (\ref{Formula_Beni}) in the limit $\epsilon\to 0$ {\em without} resorting to the explicit solution for finite $\epsilon$.\\
Our motivation in studying such limit comes from the more general project to extract as much information as possible from the still quite abstract family of two-phase solutions (see \cite{GiaPer}, \cite{GiaTop} for former geometrical results in this direction). On one side, the degeneration process that we consider allows one to find simpler solutions. Ours here are expressed in terms of elementary functions and they first appeared in the literature in 2011 due to Zakharov and Gelash. They named them "Solitons on Unstable Condensate" (\cite{ZakGel}). More recently, these ones have been obtained by Biondini and Kova\v{c}i\v{c} \cite{BioKov} after developing the inverse scattering transform for fNLS with non-zero buondary conditions at infinity. Closely related solutions were already proved to be useful in experimental studies  (see for example \cite{ChaVit}, \cite{DysTru}). On the other side, results of this type have already been employed to reconstruct (locally) the family of the two-phase solutions without algebraic geometry (\cite{GiaTop}). The method was originally applied to KdV and KP by Dubrovin in \cite{DubInv}.\\
The results we present here (Theorems \ref{th_Soliton} and \ref{th_onda_piana}) have been obtained with different methods by one of the authors in \cite{GiaTop}. 
Using in this paper only basic Riemann-Hilbert techniques we simplify significantly the calculations. Moreover, the degeneration of the two-phase solutions can be performed even without any knowledge of Riemann surfaces and theta functions. After setting our method for this specific case, we wish to apply it to more complicated ones. We are confident that this will turn out to be a valuable tool in the more general plan to provide explicit qualitative description of algebro-geometric solutions to integrable, non-linear, partial differential equations.
\vspace{1em}

We have already mentioned that Riemann-Hilbert problem \ref{Zafira_RH} describes via (\ref{Formula_Beni}) the family of two-phase solutions of fNLS. Let us characterize these ones by sketching their classical construction (\cite{BelBob}, \cite{BobKle}, \cite{DubThe}, \cite{ItsKot}. See also \cite{CalIve} which includes an interesting application of them).  

Let us keep $E_{\pm j}, \, j=1,2,3$ as defined above and consider the hyperelliptic Riemann surface 
\begin{align}
\mathcal{S}: \quad w^2 = \prod_{j=1}^{3}\lp z - E_j \rp\lp z - E_{-j} \rp.\label{Christa}
\end{align}
We will indicate with $P$ the point $\lp z,w \rp$ on $\mathcal{S}$. Its compactification has exactly two additional points, denoted by $\infty^{\pm}$, corresponding to the limit when $z$ tends to infinity and $w$ behaves like $\pm z^3$ respectively. Moreover, $\mathcal{S}$ can be endowed with the anti-holomorphic involution 
\begin{align}
\sigma: \mathcal{S}\longrightarrow \mathcal{S}, \quad\quad \sigma\lp z,w \rp = \lp \overline{z}, \overline{w} \rp.
\end{align}
Let us fix on $\mathcal{S}$ the basis in the homology
\begin{align}
\mathcal{B} = \lb a_1,a_2 ;b_1,b_2 \rb
\end{align}
indicated in figure \ref{fig_SupRiem}. Notice that this one satisfies the symmetries
\begin{align}\label{Denisa}
\sigma_{\star}\lp a_1 \rp = -a_2, \quad\quad \sigma_{\star}\lp b_1 \rp = b_2
\end{align}
where $\sigma_{\star}$ is the map induced by $\sigma$ in the homology of $\mathcal{S}$. 
\begin{figure}
\begin{center}
\includegraphics[width= 0.40\textwidth]{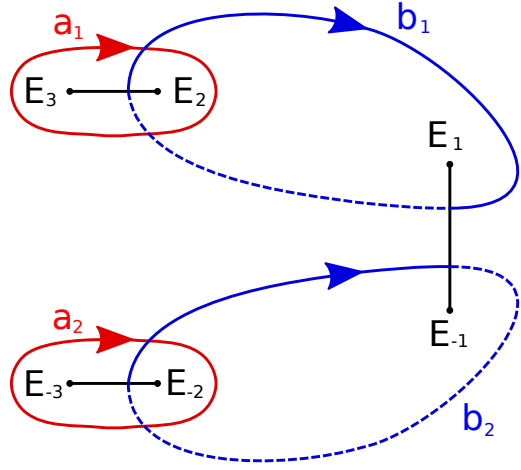}
\caption{\label{fig_SupRiem}{\it The basis in the homology $\mathcal{B}$ on the Riemann surface $\mathcal{S}$}}
\end{center}
\end{figure}
Let us denote by $\omega_1$ and $\omega_2$ the normalized holomorphic differentials satisfying 
\begin{align}
\oint_{a_j} \omega_k = 2\pi i  \delta_{jk}, \quad\quad j,k = 1,2. 
\end{align}
The period matrix $\mathbb{B}$ is defined as follows (\cite{BelBob},\cite{FarKra})
\begin{align}\label{defmatperiod}
\mathbb{B}_{jk} = \oint_{b_j}\omega_k, \quad\quad j,k= 1,2.
\end{align}
This is a $2\times 2$ symmetric matrix whose real part is negative definite. As a consequence, one can define the theta-function as follows
\begin{align}
&\Theta : \mathbb{C}^2 \rightarrow \mathbb{C}\\
& \Theta\lp \mathbf{z} \rp = \sum_{\mathbf{n}\in\mathbb{Z}^2} \exp\lb \frac{1}{2}\langle\mathbf{n},\mathbb{B}\mathbf{n}\rangle  + \langle \mathbf{n},\mathbf{z} \rangle\rb
\end{align}
for all $\mathbf{z}\in\mathbb{C}^2$. Here
\begin{align}
\langle \mathbf{x},\mathbf{y} \rangle = x_1y_1 + x_2y_2, \quad\quad \mathbf{x}, \mathbf{y}\in \mathbb{C}^2.
\end{align}
The solution $q\lp x,t;\epsilon \rp$ of the focusing NLS equation individuated by the RH-Problem \ref{Zafira_RH} via (\ref{Formula_Beni}) is then given by 
\begin{align}
\label{formula_Michela}
q\lp x,t;\epsilon \rp = A c\frac{\Theta\lp  i \mathbf{V}x  +  i \mathbf{W}t - \mathbf{D}  +\mathbf{r}  \rp}{\Theta\lp  i \mathbf{V}x  +  i \mathbf{W}t - \mathbf{D}    \rp}\exp\lp - i  E x +  i  N t \rp
\end{align}
where
\begin{align}
\mathbf{D} = \lp 2\pi  i  \rp \lp
\begin{array}{c}
\frac{\alpha}{2\pi} \\
\frac{\overline{\alpha}}{2\pi} 
\end{array}
\rp
+
\mathbb{B}\lp\begin{array}{c}
\frac{\beta}{2\pi}\\
-\frac{\overline{\beta}}{2\pi}
\end{array}\rp.
\end{align}
To determine the other parameters appearing in this formula, let us introduce the abelian integrals $\mbox{d}\Omega_1$, $\mbox{d}\Omega_2$ and $\mbox{d}\Omega_3$ determined by the following three conditions:
\begin{itemize}
\item
\begin{align}
\oint_{a_k}\mbox{d}\Omega_j = 0 \quad\quad k= 1,2; \, j= 1,2,3. 
\end{align}
\item 
\begin{align}
\begin{array}{l}
\mbox{d}\Omega_1\lp P \rp = \pm\lp 1 + o\lp 1 \rp \rp \mbox{d}z,\\
\mbox{d}\Omega_2\lp P \rp = \pm\lp 4z + o\lp 1 \rp  \rp\mbox{d}z ,\\
\mbox{d}\Omega_3\lp P \rp = \pm\lp \frac{1}{z} + o\lp 1 \rp \rp\mbox{d}z,
\end{array} \quad\quad\quad P\rightarrow \infty^{\pm}
\end{align}
\item Each of the abelian integrals $\Omega_j\lp P \rp$ has singularities only at $\infty^{\pm}$.
\end{itemize}
The numbers $E$, $N$ and $\omega_0$ are then individuated by the asymptotic expansions
\begin{align} \label{compleanno}
\begin{array}{l}
\int_{E_{-1}}^{P}\mbox{d}\Omega_1 = \pm\lp z - \frac{E}{2} + \mathcal{O}\lp 1 \rp \rp,\\
\int_{E_{-1}}^{P}\mbox{d}\Omega_2 = \pm\lp 2z^2 + \frac{N}{2} + \mathcal{O}\lp 1 \rp  \rp,\\
\int_{E_{-1}}^{P}\mbox{d}\Omega_3 = \pm\lp \ln z -\frac{1}{2}\ln \omega_0  +\mathcal{O}\lp 1 \rp \rp,
\end{array} \quad\quad\quad P\rightarrow \infty^{\pm}.
\end{align}
The two-dimensional, complex vectors $\mathbf{V}$, $\mathbf{W}$ and $\mathbf{r}$ are instead given by
\begin{align} \label{Julio}
V_j = \oint_{b_j}\mbox{d}\Omega_1, \quad\quad W_j = \oint_{b_j} \mbox{d}\Omega_2, \quad\quad r_j = -\oint_{b_j} \mbox{d}\Omega_3,  
\end{align} 
for $j=1,2$.
\\
Let us remark that symmetries (\ref{Denisa}) induce the constraints
\begin{align} \label{constraints}
V_2 = \overline{V_1}, \quad W_2 = \overline{W_1}, \quad r_2 = \overline{r_1}, \quad D_2 = -\overline{D_1} .
\end{align}
Our particular choice of a basis also implies that the quantities $E$ and $N$ are real, and that $\omega_0$ is negative. Finally 
\begin{align}
A = 2\sqrt{\omega_0}>0 
\end{align}
and $c$ is an arbitrary constant, provided that $\la c \ra=1$. 
\vspace{1em}

Our main result consists of the following two theorems:

\begin{theorem}\label{th_Soliton}
Let $\beta$ equal $\pi$. Then the solution $q\lp x,t;\epsilon \rp$ of the focusing NLS equation given by (\ref{Formula_Beni}), corresponding to the Riemann-Hilbert problem \ref{Zafira_RH}, converges, as $\epsilon$ tends to $0^+$, to the limit solution
\begin{align}\label{Krasilnikova}
q\lp x,t \rp = Ac_1\frac{
\cosh\lp \eta x+ \phi t  - i  \sigma\rp + B\cos\lp \xi x + \theta t -\alpha -  i  \rho \rp
}{
\cosh\lp \eta x+ \phi t \rp + B\cos\lp \xi x + \theta t -\alpha \rp
}
\exp\lp - i  E x +  i  N t \rp
\end{align}
The convergence of this limit is uniform on compact subsets of the $\lp x,t \rp$-plane. The value of the parameters in (\ref{Krasilnikova}) is given by (all real)
\begin{align}
&A = \Imag\lp E_1 \rp\ ,\qquad
E = 2\Real\lp E_1 \rp, \label{siitestesso}  \\ 
&N = -2\ls 2\Real\lp E_1 \rp^2 -  \Imag\lp E_1 \rp^2 \rs, \label{sempre}\\
& B = \frac{1}{\left| E_3 - E_{-3} \right|\left| E_1 - E_{-1} \right|}\lp \left| E_3 - E_1 \right| - \left| E_3 - E_{-1} \right| \rp^2, \\
& \xi + i \eta =  2 i \sqrt{  i \lp E_{-3} -E_1 \rp }\sqrt{  i \lp E_{-3} - E_{-1} \rp }, \\
& \theta + i \phi =  4 i \ls E_{-3} + \Real\lp E_1 \rp \rs\sqrt{  i \lp E_{-3} - E_{1} \rp }\sqrt{  i \lp E_{-3} - E_{-1} \rp }, \\
& \rho +i\sigma = -2\ln \ls \frac{
\sqrt{ i \lp E_{-3} - E_{-1} \rp} - \sqrt{ i \lp E_{-3} -E_1 \rp}
}{
\sqrt{ i \lp E_{-3} - E_{-1} \rp} + \sqrt{ i \lp E_{-3} -E_1 \rp}
} \rs. \\
\end{align}
Here, all the square roots are understood to assume their principal value:
\begin{align}
\sqrt{z} := \sqrt{\la z \ra} \exp\lp \frac{ i\arg{z} }{2} \rp, \quad\quad z\in \mathbb{C}\backslash\lp -\infty, 0 \rs.
\end{align}
Finally, $c_1$ is some constant complex number such that $\la c_1 \ra = 1$. 
\end{theorem}

\begin{theorem} \label{th_onda_piana}
If $p$ belongs to $\lp -\pi;\pi \rp$ then the solution $q\lp x,t; \epsilon \rp$ of the focusing NLS equation given in \ref{Formula_Beni}, corresponding to the Riemann-Hilbert problem \ref{Zafira_RH} converges, as $epsilon$ tends to $0^{+}$,  to 
\begin{align}
q\lp x,t \rp = -A c_2 \exp\lp - i  E x +  i  N t \rp.
\end{align}
The convergence is uniform on the compact subsets of the $\lp x,t \rp$-plane. Here $c_2$ is some complex number of modulus one, while the other parameters  are given by (\ref{siitestesso}-\ref{sempre}).
\end{theorem}


\begin{observation}
Let both $\alpha$ and $\beta$ equal $\pi$ and $E_3$ tend to $E_1$. Then the solution (\ref{Krasilnikova}) tends to the following limit
\begin{align}
\lim_{E_3\to E_1} q\lp x,t \rp = c_3 \cdot\small{ \lb 1 - \frac{16i\Imag\lp E_1 \rp^2 t + 4}
{ 4\Imag\lp E_1 \rp^2\ls x + 4\Real\lp E_1 \rp t \rs^2 + 16\Imag\lp E_1 \rp^4 t^2 +1 } \rb}\cdot e^{\lb -2i\Real\lp E_1 \rp x -2i\ls 2\Real\lp E_1 \rp^2 - \Imag\lp E_1 \rp^2 \rs t \rb}
\end{align}
uniformly on any compact set of the $(x,t)-$plane. Here $c_3$ is some constant complex number such that $\la c_3 \ra = 1$ This reduces, by means of a simple galilean transformation and of a scaling, to the well-known Peregrine breather:
\begin{align}
q_{Peregrine}\lp x;t \rp = \lb 1-\frac{16it + 4}{4x^2 + 16t^2 + 1} \rb \exp\lp 2it \rp
\end{align}
\end{observation}

\section{Proof of the main result} \label{Conti}

 The idea of the proof is that of approximating a preliminary modification $\bPhi\lp z; \epsilon \rp$ of the actual solution $\bPsi\lp z,\epsilon \rp$ by an appropriate  matrix $\widetilde{\bPhi}\lp z;\epsilon \rp$, independent of $\epsilon$ in a neighborhood of infinity and then giving an explicit estimate of the "error matrix" $Q\lp z;\epsilon \rp = \bPhi\lp z;\epsilon \rp
 \widetilde\bPhi(z;\epsilon)^{-1}$. 
The gist of the idea is as follows; as $\epsilon\to 0$ the underlying Riemann surface $\mathcal{S}$ degenerates to a rational curve (i.e. of genus $0$). While the exact solution $\bPsi(z;\epsilon)$ contains (in principle) Riemann Theta functions, the limiting solution and the approximation can be written in terms of completely elementary functions. The goal of the paper is thus to show that this limiting solution can be obtained {\em directly, without using any special $\Theta$ function}.

In general the idea applies also to the situation of more branch points $E_{\pm k}, \ k =1,\dots, R $ under the degeneration $|E_{\pm k} - E_{\pm(k+1)}| \to 0$ for $k = 2,3,\dots, R$ but we opted for the presentation of the simplest case not to overburden the reader with unnecessary complications which only obfuscate the underlying idea. 

Even more generally, one could decide to degenerate a smaller subset of pairs of branch points; in this case the limiting hyperelliptic curve has still a positive genus. Again, this case is not so interesting to us because the approximation of the solution would still require the use of Theta functions, in part defying our driving  purpose of simplicity and readability of the result.

\subsection{Eliminating the jumps on $\lp E_{-2};E_{-1} \rp$ and $\lp E_{1} ; E_{2} \rp$.}
We start by a preliminary step that transforms the original RHP \ref{Zafira_RH} into a more suitable (although completely equivalent) problem. This step does not involve any approximation.

Let us introduce the function 

\begin{align} \label{def_operat_radice}
R\lp z; \epsilon \rp := \lp z - E_2 \rp\lp z - E_{-1} \rp\lp z - E_{-3} \rp\sqrt{ \frac{ z - E_3 }{ z - E_2 } }\sqrt{ \frac{ z - E_{1} }{ z - E_{-1} } }\sqrt{ \frac{ z - E_{-2} }{ z - E_{-3} } }
\end{align}
where the square roots in the r.h.s. are understood to assume their principal value.
As an easy consequence, $R\lp z;\epsilon \rp$ is defined and analytic on
\begin{align}
\mathbb{C}\backslash \lb \ls E_{-3},E_{-2} \rs \cup \ls E_{-1},E_1 \rs \cup \ls E_2 , E_3 \rs \rb
\end{align}
and 
\begin{align}
\lim_{z\to \infty}R\lp z;\epsilon \rp = 1.
\end{align}
It is convenient to introduce the function 
\begin{align}
d\lp z ;\epsilon\rp := \frac{R\lp z;\epsilon \rp}{2\pi i } & \ls \hatalpha \int_{\ls E_{2};E_{3}  \rs} \frac{\mbox{d}\zeta}{\lp \zeta-z \rp R_{+}\lp \zeta;\epsilon \rp} 
-\beta \int_{\ls E_{1};E_{2}  \rs} \frac{\mbox{d}\zeta}{\lp \zeta-z \rp R_{+}\lp \zeta ;\epsilon \rp}
\right. \\ & \left.
-\beta \int_{\ls E_{-2};E_{-1}  \rs} \frac{\mbox{d}\zeta}{\lp \zeta-z \rp R_{+}\lp \zeta ; \epsilon \rp}
+ \overline{\hatalpha} \int_{\ls E_{-3};E_{-2}  \rs} \frac{\mbox{d}\zeta}{\lp \zeta-z \rp R_{+}\lp \zeta ; \epsilon\rp}
\rs
\end{align}
with  $\hatalpha$  defined by means of the equations
\begin{align}\label{trentotto_barra}
\tAbb\lp \epsilon \rp
\lp\begin{array}{c}
 \hatalpha  \\ \overline{\hatalpha}  
\end{array}\rp = 
\tBbb\lp \epsilon \rp
\lp\begin{array}{c}
 \beta  \\ \beta   
\end{array}\rp.
\end{align}
where
\begin{align}
\tAbb\lp \epsilon \rp := \ls\begin{array}{cc}
\int_{\ls E_2;E_3 \rs} \frac{\mbox{d}\zeta}{R_{+}\lp \zeta ; \epsilon \rp} & 
\int_{\ls E_{-3};E_{-2} \rs} \frac{\mbox{d}\zeta}{R_{+}\lp \zeta ; \epsilon\rp}\\
\int_{\ls E_2;E_3 \rs} \frac{\zeta\mbox{d}\zeta}{R_{+}\lp \zeta ; \epsilon \rp} & 
\int_{\ls E_{-3};E_{-2} \rs} \frac{\zeta\mbox{d}\zeta}{R_{+}\lp \zeta ; \epsilon \rp}
\end{array}\rs, \quad\quad \mbox{and} \quad\quad 
\tBbb\lp \epsilon \rp := \ls\begin{array}{cc}
\int_{\ls E_1;E_2 \rs} \frac{\mbox{d}\zeta}{R\lp \zeta;\epsilon \rp} & 
\int_{\ls E_{-2};E_{-1} \rs} \frac{\mbox{d}\zeta}{R\lp \zeta;\epsilon \rp}\\
\int_{\ls E_1;E_2 \rs} \frac{\zeta\mbox{d}\zeta}{R\lp \zeta;\epsilon \rp} & 
\int_{\ls E_{-2};E_{-1} \rs} \frac{\zeta\mbox{d}\zeta}{R\lp;\epsilon \zeta \rp}
\end{array}\rs
\end{align}
{\\
Let us first show that the matrix $\tAbb\lp \epsilon \rp$ has a finite limit as $\epsilon\to 0$. Indeed, consider for example the entry (1,1). By the change of variable $\zeta = E_3 + \epsilon t$ 
one obtains
\begin{align}
\int_{\ls E_2,E_3\rs}\frac{\mbox{d}\zeta}{R_{+}\lp \zeta ; \epsilon \rp} = - \int_{0}^{1}\frac{\epsilon\mbox{d}t}{R_{+}\lp E_3 + \epsilon t ;\epsilon \rp}.
\end{align}
On the other side, it is easy to obtain that
\begin{align}
\frac{1}{R_{+}\lp E_3 + \epsilon t ; \epsilon\rp} = \frac{i}{\epsilon\sqrt{t\lp 1-t \rp}}\lb \frac{1}{\lp E_3 - E_{-1} \rp\lp E_3 - E_{-3} \rp} \sqrt{ \frac{E_3 - E_{-1}}{ E_{3} - E_1 } } + r\lp t;\epsilon \rp\rb
\end{align}
where $r\lp t, \epsilon \rp\leq M\epsilon$ for some positive constant $M$ independent of $t\in\ls 0,1 \rs$ and $\epsilon$ sufficiently small. This yields immediately
\begin{align}
\int_{\ls E_2,E_3 \rs}\frac{\mbox{d}\zeta}{R_{+}\lp \zeta ; \epsilon  \rp} &= \frac{-i}{\lp E_3 - E_{-1} \rp\lp E_3 - E_{-3} \rp}\sqrt{\frac{E_3 - E_{-1}}{E_3 - E_1}}\int_0^1 \frac{\mbox{d}t}{\sqrt{t\lp 1-t \rp}} + \mathcal{O}\lp \epsilon \rp\\
& = \frac{-i\pi}{\lp E_3 - E_{-1} \rp\lp E_3 - E_{-3} \rp}\sqrt{\frac{E_3 - E_{-1}}{E_3 - E_1}} + \mathcal{O}\lp \epsilon \rp
\end{align}
With analogous considerations for the other entries of the matrix one obtains
\begin{align}\label{Manon_Lescaut}
\tAbb\lp \epsilon \rp \,\,=\,\,
\ls\begin{array}{cc}
\frac{-i\pi}{\lp E_3 - E_{-1} \rp\lp E_3 - E_{-3} \rp}\sqrt{\frac{E_3 - E_{-1}}{E_3 - E_1}}  & 
\frac{i\pi}{\lp E_{-3} - E_{1} \rp\lp E_3 - E_{-3} \rp}\sqrt{\frac{E_{-3} - E_1}{E_{-3} - E_{-1}}} \\
\frac{-i\pi E_{3}}{\lp E_3 - E_{-1} \rp\lp E_3 - E_{-3} \rp}\sqrt{\frac{E_{3} - E_{-1}}{E_3 - E_1}}  &
\frac{i\pi E_{-3}}{\lp E_{-3} - E_{1} \rp\lp E_{3} - E_{-3} \rp}\sqrt{\frac{E_{-3} - E_1}{E_{-3} - E_{-1}}} 
\end{array}\rs+ \mathcal{O}\lp \epsilon \rp, \quad\quad \epsilon\to 0. 
\end{align}
This also implies that the determinant of $\tAbb\lp \epsilon \rp$ is different than zero for sufficiently small $\epsilon$. This fact is actually more general  and holds for every triple of distinct $E_1,E_2$ and $E_3$ with positive imaginary parts (see \cite{FarKra}). Using (\ref{Manon_Lescaut}), elementary manipulations yield
\begin{align}\label{Audrine}
\tAbb\lp \epsilon \rp^{-1}\tBbb\lp \epsilon \rp = 
\ls \mbox{diag}\lp\begin{array}{cc}  
-\frac{E_3 - E_{-1}}{i\pi}\sqrt{\frac{E_3 - E_1}{E_3 - E_{-1}}} &
-\frac{E_{-3}-E_1}{i\pi}\sqrt{\frac{ E_{-3} - E_{-1} }{ E_{-3} - E_1 }} 
\end{array} \rp + \mathcal{O}\lp \epsilon \rp \rs\cdot\small{
\ls\begin{array}{cc}
\int_{\ls E_1,E_2 \rs}\omega_1\lp \zeta ; \epsilon \rp\mbox{d}\zeta &
\int_{\ls E_{-2},E_{-1} \rs} \omega_1\lp \zeta ; \epsilon \rp\mbox{d}\zeta\\
\int_{\ls E_1,E_2 \rs}\omega_2\lp \zeta ; \epsilon \rp\mbox{d}\zeta &
\int_{\ls E_{-2},E_{-1} \rs} \omega_2\lp \zeta ; \epsilon \rp\mbox{d}\zeta
\end{array}\rs}
\end{align}
where the functions $\omega_1\lp \zeta ; \epsilon \rp$ and $\omega_2\lp \zeta; \epsilon \rp$ are defined as follows
\begin{align}
\omega_1\lp \zeta ; \epsilon \rp = \frac{\zeta - E_{-3}}{R\lp \zeta;\epsilon \rp}, \quad\quad 
\omega_2\lp \zeta ; \epsilon \rp = \frac{\zeta - E_3}{R\lp \zeta;\epsilon \rp}.
\end{align}
Taking the limit inside the integral and performing an elementary integration one obtains
\begin{align}
\int_{\ls E_{-2},E_{-1} \rs}\omega_1\lp \zeta ; \epsilon \rp \mbox{d}\zeta & \to
\int_{\ls E_{-2},E_{-1} \rs}\sqrt{\frac{ \zeta - E_{-1} }{ \zeta - E_{1} }}\frac{\mbox{d}\zeta}{\lp \zeta - E_3 \rp\lp \zeta - E_{-1} \rp}\\
& = -\frac{1}{\lp E_{3} - E_{-1} \rp}\sqrt{\frac{E_3 - E_{-1}}{E_3 - E_{1}}}
\ln \ls \frac{ \lp \la E_1 - E_3 \ra - \la E_3-E_{-1} \ra \rp^2 }{ \la E_1 - E_{-1} \ra\la E_3 - E_{-3} \ra } \rs, \quad\quad \epsilon \to 0.
\end{align}
On the other hand, introducing the notation
\begin{align}
f\lp \zeta \rp := \frac{1}{ \zeta - E_{-1} } \sqrt{\frac{\zeta - E_{-1}}{\zeta - E_1}} \sqrt{\frac{\zeta -E_{-3}}{\zeta - E_{-2}}} 
\end{align}
one has
\begin{align}
\int_{\ls E_1,E_2 \rs} \omega_1\lp \zeta;\epsilon \rp\mbox{d}\zeta &= \int_{\ls E_1,E_2 \rs} \frac{f\lp \zeta,\epsilon \rp}{\zeta-E_2}\sqrt{\frac{\zeta - E_2}{\zeta - E_3}} \mbox{d}\zeta\\
&=  f\lp E_2 \rp \int_{\ls E_1,E_2 \rs}\frac{1}{\zeta - E_2}\sqrt{\frac{\zeta -E_2}{\zeta -E_3}}\mbox{d}\zeta + \int_{\ls E_1, E_2 \rs} \frac{f\lp \zeta,\epsilon \rp - f\lp E_2, \epsilon \rp }{\zeta -E_2}\sqrt{\frac{\zeta-E_2}{\zeta - E_3}}\mbox{d}\zeta \\
& = f\lp E_2,\epsilon \rp \left.\ln \ls t + \lp t-1 \rp\sqrt{\frac{t+1}{t-1}} \rs\right|^{t=1}_{t=\frac{2\lp E_1-E_3 \rp -\epsilon}{\epsilon}} + \int_{\ls E_1,E_3 \rs} \frac{ f\lp \zeta,0 \rp -f\lp E_3;0 \rp}{\zeta-E_3}\mbox{d}\zeta + \mathcal{O}\lp \epsilon \rp\\
& = \frac{1}{E_3-E_{-1}}\sqrt{\frac{E_3-E_{-1}}{E_{3}}} \lb \ln \epsilon - \ln \ls 4\lp E_1-E_3 \rp \rs \rb + \int_{\ls E_1,E_3 \rs} \frac{ f\lp \zeta,0 \rp -f\lp E_3;0 \rp}{\zeta-E_3}\mbox{d}\zeta + \mathcal{O}\lp \epsilon\ln \epsilon \rp.  
\end{align}
A direct integration of the last integral and a bit of algebra yield then
\begin{align}
\int_{\ls E_1,E_2 \rs} \omega_1\lp \zeta ; \epsilon \rp \mbox{d}\zeta = \frac{1}{E_3-E_{-1}}\sqrt{\frac{E_3-E_{-1}}{E_3-E_1}}\lb \ln \epsilon + \ln \ls \frac{E_{-1}-E_1}{16\lp E_3 - E_1 \rp\lp E_3 - E_{-1} \rp} \rs + \mathcal{O}\lp \epsilon\ln \epsilon \rp  \rb
\end{align}
After analogous considerations for the other integrals appearing in (\ref{Audrine}), one obtains
\begin{align}
\tAbb\lp \epsilon \rp^{-1}\tBbb\lp \epsilon \rp = \ls
\begin{array}{cc}
-\frac{1}{i\pi}\lb \ln \epsilon + \ln \ls \frac{E_{-1}-E_1}{16\lp E_3-E_1 \rp\lp E_3-E_{-1} \rp} \rs \rb &
\frac{1}{i\pi}\ln \ls \frac{\lp \la E_1-E_3 \ra - \la E_{3} - E_{-1} \ra \rp^2}{\la E_{1}-E_{-1} \ra\la E_3-E_{-3} \ra} \rs\\
-\frac{1}{i\pi}\ln \ls \frac{\lp \la E_1 -E_3 \ra- \la E_3-E_{-1} \ra \rp^2}{\la E_1 -E_{-1} \ra\la E_{3}-E_{-3} \ra} \rs &
-\frac{1}{i\pi} \lb \ln \epsilon+\ln \ls \frac{E_1-E_{-1}}{16\lp E_{-3}-E_1 \rp\lp E_{-3}-E_{-1} \rp} \rs \rb
\end{array}\rs + \mathcal{O}\lp \epsilon\ln \epsilon \rp.
\end{align}
From (\ref{trentotto_barra}) and elementary algebra then,
\begin{align}\label{giorgia_uno}
\hat{\alpha}= \frac{i\beta}{\pi}\ln \lp \frac{\epsilon}{4i\mathcal{H}} \rp + \mathcal{O}\lp \epsilon\ln \epsilon \rp
\end{align}
where $\mathcal{H}$ is the following, $\epsilon$-independent constant:
\begin{align}\label{giorgia_due}
\mathcal{H} = \frac{\lp E_3-E_1 \rp\lp E_3-E_{-1} \rp\lp \la E_1-E_{3} \ra - \la E_3-E_{-1} \ra \rp^2}{2\mbox{Im}\lp E_1 \rp^2\mbox{Im}\lp E_3 \rp}.
\end{align}
The following proposition summarizes  the analytic properties of $d(z;\epsilon)$ that will be used in the sequel.}
\begin{proposition} \label{Jenny}
The function $d\lp z;\epsilon \rp$ is analytic on 
\begin{align}
\mathbb{C}\backslash \lb \ls E_{-3} ; E_{-2} \rs \cup  \ls E_{-2} ; E_{-1} \rs \cup  \ls E_{1} ; E_{2} \rs \cup  \ls E_2 ; E_3 \rs \rb.
\end{align}
It satisfies the following boundary values relations:
\begin{align}
\left\{
\begin{array}{ll}
d_{+}\lp z ;\epsilon\rp + d_{-}\lp z ;\epsilon\rp = \hatalpha(\epsilon) & z\in \lp E_2;E_3   \rp \\
d_{+}\lp z ;\epsilon\rp - d_{-}\lp z ;\epsilon\rp = -\beta & z\in  \lp E_1;E_2  \rp \\
d_{+}\lp z ;\epsilon\rp - d_{-}\lp z;\epsilon \rp = -\beta   & z\in  \lp E_{-2};E_{-1}   \rp \\
d_{+}\lp z;\epsilon \rp + d_{-}\lp z;\epsilon \rp = \overline{\hatalpha(\epsilon)} & z\in \lp E_{-3};E_{-2}   \rp.
\end{array}
\right. .
\end{align}
Moreover  the limit below exists, finite and real
\begin{align}
d_{\infty} (\epsilon):= \lim_{z\rightarrow \infty} d\lp z ;\epsilon\rp \in \mathbb R.
\end{align}
 Letting then $\epsilon$ tend to $0^{+}$, also $\lim_{\epsilon\rightarrow 0^+} d_{\infty}\lp \epsilon \rp $
exists and it is finite.
\end{proposition}
\begin{proof}
Let us first fix $\epsilon>0$. One has 
\begin{align} \label{Beppe}
\int_{\ls E_{-3} ; E_{-2} \rs}\frac{\mbox{d}\zeta}{\lp \zeta-z \rp R_{+}\lp \zeta; \epsilon \rp} =& -\frac{1}{z}\int_{\ls E_{-3} ; E_{-2} \rs}\frac{\mbox{d}\zeta}{ R_{+}\lp \zeta ; \epsilon \rp}
 -\frac{1}{z^2}\int_{\ls E_{-3} ; E_{-2} \rs}\frac{\zeta\mbox{d}\zeta}{ R_{+}\lp \zeta ; \epsilon \rp} \\ &-\frac{1}{z^3}\int_{\ls E_{-3} ; E_{-2} \rs}\frac{\zeta^2\mbox{d}\zeta}{ R_{+}\lp \zeta ; \epsilon \rp} + \mathcal{O}\lp z^4 \rp, \quad \quad z\rightarrow \infty.
\end{align}
The same expansion holds when $\ls E_{-3};E_{-3} \rs$ with $\ls E_{-2};E_{-1} \rs$, $\ls E_1;E_2 \rs$ or $\ls E_2;E_3 \rs$. On the other side,
\begin{align}
R\lp z;\epsilon \rp\sim z^3, \quad\quad z\rightarrow \infty.
\end{align}
So $d_{\infty}(\epsilon)$ exists and it is finite if 
\begin{subequations}\label{dentista}
\begin{align}
\hatalpha(\epsilon) \int_{\ls E_2;E_3 \rs}\frac{\mbox{d}\zeta}{R_{+}\lp \zeta ; \epsilon \rp} - \beta \int_{\ls E_1;E_2 \rs}\frac{\mbox{d}\zeta}{R_{+}\lp \zeta ; \epsilon \rp} - \beta \int_{\ls E_{-2};E_{-1} \rs}\frac{\mbox{d}\zeta}{R_{+}\lp \zeta ; \epsilon \rp} + \overline{\hatalpha(\epsilon)} \int_{\ls E_{-3};E_{-2} \rs}\frac{\mbox{d}\zeta}{R_{+}\lp \zeta ; \epsilon \rp} = 0
\end{align}
and
\begin{align}
\hatalpha(\epsilon) \int_{\ls E_2;E_3 \rs}\frac{\zeta\mbox{d}\zeta}{R_{+}\lp \zeta;\epsilon \rp} - \beta \int_{\ls E_1;E_2 \rs}\frac{\zeta\mbox{d}\zeta}{R_{+}\lp \zeta ; \epsilon \rp} - \beta \int_{\ls E_{-2};E_{-1} \rs}\frac{\zeta\mbox{d}\zeta}{R_{+}\lp \zeta ; \epsilon \rp} + {\hatalpha(\epsilon)} \int_{\ls E_{-3};E_{-2} \rs}\frac{\zeta\mbox{d}\zeta}{R_{+}\lp \zeta ; \epsilon \rp} = 0.
\end{align}
\end{subequations}
These ones hold because they are equivalent to equations (\ref{trentotto_barra}). Again from (\ref{Beppe}),
\begin{align} \label{Dorina}
d_{\infty}(\epsilon) = \hatalpha(\epsilon) \int_{\ls E_2;E_3 \rs}\frac{\zeta^2\mbox{d}\zeta}{R_{+}\lp \zeta ; \epsilon\rp} - \beta \int_{\ls E_1;E_2 \rs}\frac{\zeta^2\mbox{d}\zeta}{R_{+}\lp \zeta ; \epsilon \rp} - \beta \int_{\ls E_{-2};E_{-1} \rs}\frac{\zeta^2\mbox{d}\zeta}{R_{+}\lp \zeta ; \epsilon \rp} + \hatalpha(\epsilon) \int_{\ls E_{-3};E_{-2} \rs}\frac{\zeta^2\mbox{d}\zeta}{R_{+}\lp \zeta ; \epsilon\rp}.
\end{align}
This one turns out to be real because $R_{+}\lp \overline{z} ; \epsilon \rp$ equals $R_{+
}\lp z ; \epsilon\rp$ in view of the Schwarz reflection principle. In order to prove that the limit of $d_\infty$ as $\epsilon$ tends to zero exists and it is finite, it is sufficient to expand each of the integrals in (\ref{Dorina}) up to leading term and use (\ref{giorgia_uno}). 
\end{proof}
We now use $d(z;\epsilon)$ in order to ``normalize'' the solution of the RHP \ref{Zafira_RH}; to this extent we  introduce  

\begin{align} \label{Def_Phi}
\bPhi \lp z, \epsilon \rp := e^{ i  d_{\infty}(\epsilon)\sigma_3} \bPsi\lp z; \epsilon \rp e^{- i  d\lp z ;\epsilon \rp\sigma_3}.
\end{align}
The dependence of $\bPhi$ from $x$ and $t$ is understood although not explicitly indicated. For any fixed ($x$, $t$ and) $\epsilon>0$, $\bPsi\lp z; \epsilon \rp$ solves Riemann-Hilbert problem \ref{Zafira_RH} if and only if $\bPhi\lp z, \epsilon \rp $ satisfies the following one:
\begin{RHP}
\label{TroppeDonne}
Determine a matrix-valued function 
\begin{align}
\bPhi : \mathbb{C}\backslash \lb \ls E_{-3}  ; E_{-2} \rs \cup \ls E_{-1}  ; E_{1} \rs \cup \ls E_{2}  ; E_{3} \rs \rb \longrightarrow Mat\lp 2\times 2,\mathbb{C} \rp
\end{align}
such that
\begin{enumerate}
\item[\textbf{i}] $\bPhi $ is analytic on its domain.\\
\item[$\mathbf{ii}$] For every point $z$ of $\lp E_{-3};E_{-2} \rp \cup \lp E_{-1};E_{1} \rp \cup \lp E_{2};E_{3} \rp$
one has 
\begin{align}
\bPhi _{+}\lp z \rp = \bPhi _{-}\lp z \rp M_{\bPhi }\lp z \rp
\end{align}
where 
\begin{align}\label{saltoPhi}
M_{\bPhi }= \left\{\begin{array}{cc}
\exp\left\{-i\ls \tf + \frac{\alpha -\hat{\alpha}(\epsilon)}{2} \rs\stre\right\}\lp -i\sdue\rp \exp\left\{i\ls \tf +\frac{\alpha-\hat{\alpha}(\epsilon)}{2} \rs\stre\right\} & z\in \lp E_2 ; E_3 \rp\\
\exp\left\{-i \tf  \stre\right\}\lp - i\sdue\rp \exp\left\{i \tf \stre\right\} & z\in \lp E_{-1} ; E_1 \rp\\
\exp\left\{-i\ls \tf + \frac{\alpha -\overline{\hat{\alpha}(\epsilon)}}{2} \rs\stre\right\}\lp -i\sdue\rp \exp\left\{i\ls \tf + \frac{\alpha - \overline{\hat{\alpha}(\epsilon)}}{2} \rs\stre\right\} & z\in \lp E_{-3} ; E_{-2} \rp\\
\end{array} \right.
\end{align}
\item[$\mathbf{iii}$] The following asymptotic behaviour holds:
\begin{align}
\bPhi \lp z \rp = \mathbf 1 + \mathcal{O}\lp \frac{1}{z} \rp, \quad\quad \mbox{as} \quad z\rightarrow \infty.
\end{align}
\item[$\mathbf{iv}$] The growth condition
\begin{align}
\bPhi \lp z \rp = \mathcal{O}\lp \left| z - E_j \right|^{-\frac{1}{2}} \rp \quad\quad \mbox{as} \quad z\rightarrow E_j
\end{align}
is satisfied for $j= \pm 1, \pm 2, \pm 3$.
\end{enumerate}
\end{RHP}

\begin{remark}
The simplification of the RHP \ref{TroppeDonne} relative to the RHP \ref{Zafira_RH} is that the diagonal jump matrices in the latter problem have been eliminated.
It is the solution of the RHP \ref{TroppeDonne}  that will be approximated in the small $\epsilon$ regime.
\end{remark}

\subsection{Approximation  for $\bPhi\lp z;\epsilon\rp $ as $\epsilon \to 0$}
In order to approximate solutions $q\lp x,t;\epsilon \rp$ (\ref{Formula_Beni}) of fNLS, we construct a uniform approximation of $\bPhi\lp z;\epsilon \rp$ (and so of $\bPsi\lp z;\epsilon \rp$) on $\C$
\vspace{1em}

Let us introduce two simply connected regions $\mathbb{D}_{1}$ and $\mathbb{D}_{-1}$ surrounding the segments $\ls  E_{-3};E_{-2} \rs$ and $\ls E_2;E_3 \rs$, as indicated in figure \ref{fig_RHPRegions}. Let us also put 
\begin{align}
\mathbb{D}_{0}:= \mathbb{C}\backslash \lb \mathbb{D}_{1}\cup \mathbb{D}_{-1}\cup \ls E_{-1};E_{1} \rs \rb.
\end{align}
We will now define an approximation $\widetilde{\bPhi }\lp z;\epsilon \rp$ for $\bPhi\lp z;\epsilon \rp $ as $\epsilon\to 0^{+}$, piecewise on each of these regions.
\begin{figure}
\resizebox{0.5\textwidth}{!}{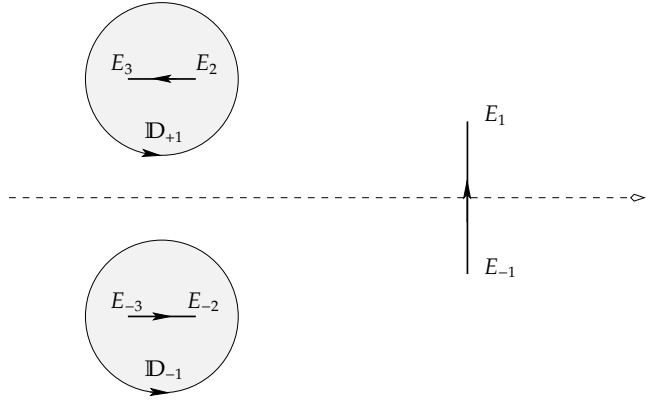}
\caption{\label{fig_RHPRegions}{\it The regions $\mathbb D_{\pm 1}$ surrounding the small spectral cuts $[E_{\pm 3}, E_{\pm 2}]$ with $E_{\pm 2} = E_{\pm 3} + \epsilon$ and $\epsilon>0$, small. The approximation $\widehat \bPhi(z;\epsilon)$ is constructed as a piecewise analytic function.}}
\end{figure}

 {\bf Step 0}: We start by defining the $\epsilon$-independent approximation in the region $\mathbb D_0$. Let us introduce the matrix-valued function 
\begin{align}
\mathcal{N}_{0}\lp z \rp =
\ls \begin{array}{cc}
\frac{ \mu_{0}\lp z \rp + \mu_{0}\lp z \rp^{-1} }{2} & i\frac{ \mu_{0}\lp z \rp - \mu_{0}\lp z \rp^{-1}  }{2}\\
\frac{  \mu_{0}\lp z \rp - \mu_{0}\lp z \rp^{-1} }{2i} & \frac{ \mu_{0}\lp z \rp + \mu_{0}\lp z \rp^{-1}  }{2}
\end{array}\rs = 
\left(\frac {z-E_1}{z-E_{-1}}\right)^ {-\frac {\sigma_2}4 } ,\quad\quad z\in\mathbb{C}\backslash\ls E_{-1};E_{1} \rs
\end{align}
where
\begin{align}
\mu_{0}\lp z \rp = \sqrt[4]{\frac{z-E_1}{z-E_{-1}}}, \quad\quad z\in\mathbb{C}\backslash \ls E_{-1};E_{1} \rs.
\end{align}
The function $\mu_{0}\lp z \rp$ is understood to be analytic and single-valued on its domain and to tend to 1 as $z$ tends to infinity. Let us also introduce the scalar function
\begin{align}
d_{0}\lp z;x,t \rp := \theta\lp z;x,t \rp -\ls 2t\lp z + \frac{E_1 + E_{-1}}{2} \rp + x \rs R_0\lp z \rp, 
\end{align}
where $\theta$ is defined in (\ref{def_theta}) and
\begin{align}
R_0\lp z \rp := \lp z - E_{-1} \rp\sqrt{\frac{z-E_1}{z-E_{-1}}}
\end{align}
The function $R_0$ is also understood to be analytic and single-valued on $\mathbb{C}\backslash \ls E_{-1};E_{1} \rs$ and to asymptotically behave like $z$ as $z$ tends to infinity.\\
Using these objects we define
\begin{align} \label{Def_psi_zero}
\psi_0\lp z \rp = e^{- i  d_{0,\infty}\sigma_3}\mathcal{N}_{0}\lp z \rp e^{ i  d_0\lp z \rp\sigma_3}, \quad\quad z\in \mathbb{C}\backslash\ls E_{-1};E_1 \rs
\end{align}
where
\begin{align}
d_{0,\infty} := \lim_{z\rightarrow \infty} d_0\lp z \rp = \ls \frac{\lp E_1 + E_{-1} \rp^2}{2} + \frac{\lp E_{1} - E_{-1} \rp^2}{4} \rs t + \frac{\lp E_1 + E_{-1} \rp}{2} x.
\end{align}
By means of a straightforward though long calculation one can prove the following 
\begin{proposition}
The function $\psi_0\lp z \rp$ is analytic on $\mathbb{C}\backslash\ls E_{-1};E_1 \rs$. For every point $z$ of $\ls E_{-1};E_1 \rs$ one has
\begin{align}
\psi_{0,+}\lp z \rp = \psi_{0,-}\lp z \rp M_{\psi_0}\lp z \rp, \quad\quad z\in\ls E_{-1};E_1 \rs
\end{align}
where
\begin{align}
M_{\psi_0}\lp z \rp = e^{- i \theta\lp z;x,t \rp\sigma_3}\lp - i \sigma_2 \rp e^{ i \theta\lp z;x,t \rp\sigma_3}, \quad\quad z \in \ls E_{-1};E_{1} \rs.
\end{align}
Moreover, 
\begin{align}
\psi_0\lp z \rp &= \mathbf 1 + \mathcal{O}\lp \frac{1}{z} \rp, \quad\quad z\rightarrow \infty.\\
\psi_0 (z) & = \mathcal O( |z - E_{\pm 1})|^{-\frac 1 4})\ ,\ \ z \to E_{\pm 1}.
\label{asymp0}
\end{align}
\end{proposition}
\textbf{Step 1}: We define an ($\epsilon$-dependent) approximation of $\bPhi\lp z;\epsilon\rp$ on $\mathbb{D}_{\pm1}$.
We define
\begin{align}\label{Le}
\mathcal{N}_{\pm 1}\lp z ; \epsilon \rp := 
\ls\begin{array}{cc}
\frac{\mu_{\pm 1}\lp z ; \epsilon \rp + \mu_{\pm 1}\lp z ; \epsilon \rp^{-1}}{2} & i\frac{\mu_{\pm 1}\lp z ; \epsilon \rp - \mu_{\pm 1}\lp z ; \epsilon \rp^{-1}}{2}\\
\frac{\mu_{\pm 1}\lp z ; \epsilon\rp - \mu_{\pm 1}\lp z ; \epsilon \rp^{-1}}{2i} & \frac{\mu_{\pm 1}\lp z ; \epsilon\rp + \mu_{\pm 1}\lp z ; \epsilon\rp^{-1}}{2}
\end{array}\rs, \quad\quad z\in \mathbb{C}\backslash\ls E_{\pm 2}; E_{\pm 3} \rs
\end{align}
where
\begin{align}\label{na}
\mu_{\pm 1}\lp z; \epsilon \rp := \lp \frac{z-E_{\pm 3}}{z- E_{\pm 2}} \rp^{\pm \frac 1 4},       \quad\quad z\in \mathbb{C}\backslash \ls E_{\pm 2}; E_{\pm 3} \rs.
\end{align}
Also in this case, the functions $\mu_{\pm 1}\lp z ; \epsilon\rp$ are understood to be analytic and single-valued on their domain, and they tend to one as $z$ tends to $\infty$.
We then put 
\begin{align}
\psi_1\lp z ;\epsilon\rp := e^{ -i\ls \theta\lp z;x,t \rp + \frac{\alpha-\hat{\alpha (\epsilon)}}{2} \rs\sigma_3 }\cdot \mathcal{N}_1\lp z;\epsilon \rp\cdot
 e^{i\ls \theta\lp z;x,t \rp + \frac{\alpha -\hatalpha(\epsilon)}{2} \rs\sigma_3},  
\end{align}
for $z\in \mathbb{C}\backslash\ls E_2;E_3 \rs$, and
\begin{align}
\psi_{-1}\lp z;\epsilon \rp := e^{ -i\ls \theta\lp z;x,t \rp + \frac{\alpha-\overline{\hat{\alpha}(\epsilon)}}{2} \rs\sigma_3 }\cdot \mathcal{N}_{-1}\lp z;\epsilon \rp\cdot
 e^{i\ls \theta\lp z;x,t \rp + \frac{\alpha -\overline{\hatalpha (\epsilon)} }{2} \rs\sigma_3},  
\end{align}
for $z\in \mathbb{C}\backslash\ls E_{-3};E_{-2} \rs$. 
With these objects one can finally define a candidate to approximate $\bPhi $ as follows:
\begin{align}\label{Def_Phi_tilde}
\widetilde{\bPhi }\lp z;\epsilon \rp := \lb
\begin{array}{ll}
\psi_0\lp z \rp & z\in \mathbb{D}_0\\
\psi_0\lp z \rp\psi_1\lp z;\epsilon \rp & z\in \mathbb{D}_1\\
\psi_0\lp z \rp\psi_{-1}\lp z;\epsilon \rp & z\in \mathbb{D}_{-1}.
\end{array}
\right.
\end{align}
\begin{remark}\label{Cristina}
$\widetilde{\bPhi }$ has the same jumps as $\bPhi $ on $\lp E_{-3} ; E_{-2}\rp \cup  \lp E_{-1};E_{1} \rp \cup \lp E_{2} ;E_{3} \rp$
but it has also some additional jumps on $\partial \mathbb{D}_{\pm 1}$.
\end{remark}

\subsection{Estimating the error}
We now discuss whether $\widetilde{\bPhi }$ is a good approximation for $\bPhi $ as $\epsilon \to 0^+$.  To  this purpose we the error matrix
\begin{align}
Q\lp z;\epsilon \rp := \bPhi \lp z;\epsilon \rp\widetilde{\bPhi }\lp z;\epsilon \rp^{-1}.
\label{81}
\end{align}
We anticipate already here that $Q(z;\epsilon)$ is {\bf not} close to the identity matrix $\mathbf 1$ as $\epsilon\to 0$ {\em in general}. Indeed this happens only if $\beta \in (-\pi, \pi)$ but not if $\beta = \pm \pi$ (See Prop. \ref{Ali} below).

As a consequence of its definition \eqref{81} the matrix $Q(z; \epsilon)$ has a discontinuity on  $\partial \mathbb D_\pm$:
\begin{align}
Q_{+}\lp z;\epsilon \rp = \bPhi _{+}\lp z;\epsilon \rp \widetilde{\bPhi }_{+}\lp z;\epsilon \rp^{-1} = \bPhi _{-}\lp z;\epsilon \rp\widetilde{\bPhi }_{-}\lp z;\epsilon \rp^{-1} \widetilde{\bPhi }_{-}\lp z;\epsilon \rp\widetilde{\bPhi }_{+}\lp z;\epsilon \rp^{-1} = Q_{-}\lp z;\epsilon \rp \ls \widetilde{\bPhi }_{-}\lp z;\epsilon \rp\widetilde{\bPhi }_{+}\lp z;\epsilon \rp^{-1} \rs
\end{align}
Consequently we deduce that $Q(z;\epsilon)$ is a piecewise analytic matrix--valued function on $\mathbb C\setminus \partial \mathbb D_{+1} \cup\partial \mathbb D_{-1}$ and 
\begin{align}
Q_+(z;\epsilon) &= Q_-(z;\epsilon) M_Q(z;\epsilon)\qquad  z\in \partial \mathbb{D}_{\pm 1}
\\
Q(z;\epsilon) &= \mathbf 1 + \mathcal O(z^{-1})\ ,\ \ \ \ |z|\to \infty\\
&M_{Q}\lp z;\epsilon \rp = \widetilde{\bPhi }_{-}\lp z ;\epsilon \rp \widetilde{\bPhi }_{+}\lp z;\epsilon \rp^{-1}= \psi_0\lp z \rp\psi_{\pm 1}\lp z ;\epsilon\rp^{-1}\psi_0\lp z \rp^{-1}
\end{align}
 The behaviour of $M_Q\lp z;\epsilon \rp$ as $\epsilon$ tends to zero depends on $\beta$, as stated in the following 
\begin{proposition} \label{Ali}
 { \bf 1]} Let $\beta$ belong to $( -\pi; \pi )$. Then, uniformly in $z\in \partial \mathbb D_{\pm 1}$, we have 
\begin{align}\label{Maldonado}
\lim_{\epsilon\rightarrow 0^{+}}M_{Q}\lp z ;\epsilon\rp = \mathbf 1.
\end{align}
{\bf 2]} Let instead $\beta$ equal $\pi$. Then, uniformly w.r.t. $z$
\begin{align} \label{Mo}
\lim_{\epsilon\rightarrow 0^{+}} M_{Q}\lp z;\epsilon \rp = \mathbf 1 +  \frac{\mathcal{H}e^{-i \ls 2\theta\lp z; x,t \rp + \alpha \rs}}{\lp z-E_3 \rp}\psi_0\lp z\rp  \sigma_+ \psi_0\lp z \rp^{-1} =: M_{Q}(z;0), \quad\quad z\in \partial \mathbb{D}_1
\\
\label{lina}
\lim_{\epsilon\rightarrow 0^+} M_{Q}\lp z;\epsilon \rp =\mathbf 1 -\frac{\overline{\mathcal{H}}e^{i\ls 2\theta\lp z;x,t \rp + \alpha \rs}}{\lp z-E_{-3} \rp} \psi_0\lp z \rp
\sigma_- \psi_0\lp z \rp^{-1}=: M_{Q}(z;0), \quad\quad z\in \partial \mathbb{D}_{-1}.
\end{align}
where $\sigma_+ = \le[\begin{array}{cc}
0& 1\\
0&0
\end{array}\ri],\ \ \sigma_- = \le[\begin{array}{cc}
0& 0\\1&0
\end{array}\ri].$
and $\mathcal H$   is in  (\ref{giorgia_due}).
In particular the convergence of the matrices $M_Q(z;\epsilon)$ is also in all the $L^p$ spaces, $1 \leq p \leq \infty$. 
\end{proposition}
\begin{proof}
Let $z$ belong to $\partial \mathbb{D}_1$. 
Inspection of \eqref{Le} using \eqref{na}, \eqref{movimento} shows that 
\begin{equation}
\mathcal N_1(z;\epsilon) = \left[
\begin{array}{cc}
1 + \mathcal O(\epsilon^2)  &  \frac{ i \epsilon }{4(z-E_3)} + \mathcal O(\epsilon^2)\\
\frac{ -i \epsilon }{4(z-E_3)} + \mathcal O(\epsilon^2) & 1+ \mathcal O(\epsilon^2)
\end{array}
\right]
\end{equation}
and hence 
\begin{align}
M_Q(z;\epsilon) &= 
\psi_0(z)  \left[
\begin{array}{cc}
1 + \mathcal O(\epsilon^2)  &  \frac{ -i \epsilon {\rm e}^{-2i\theta - i \alpha + i\hat \alpha(\epsilon)} }{4(z-E_3)} + \mathcal O(\epsilon^2)\\
\frac{ i \epsilon  {\rm e}^{2i\theta + i\alpha -i \hat \alpha(\epsilon)} }{4(z-E_3)} + \mathcal O(\epsilon^2) & 1+ \mathcal O(\epsilon^2)
\end{array}
\right]\psi_0(z)^{-1} 
\\&=
\mathbf 1  +  \psi_0(z)\left[
\begin{array}{cc}
\mathcal O(\epsilon^2)  & 
\le(\frac   \epsilon {4i{\mathcal H}}\ri)^{  -\frac \beta \pi}
 \frac{ -i \epsilon {\rm e}^{-2i\theta  - i \alpha} }{4(z-E_3)} (1 + \mathcal O(\epsilon\ln \epsilon))+ \mathcal O(\epsilon^2)\\
\le(\frac   \epsilon {4i \mathcal H }\ri)^{ \frac \beta \pi}\frac{ i \epsilon  {\rm e}^{2i\theta + i\alpha } }{4(z-E_3)}(1 + \mathcal O(\epsilon\ln \epsilon)) + \mathcal O(\epsilon^2) & \mathcal O(\epsilon^2)
\end{array}
\right]\psi_0(z)^{-1} 
\label{conju}
\end{align}
%
 where we have used  \eqref{giorgia_uno}, so that  
\begin{align}
\exp \lp \pm i \hat\alpha(\epsilon)\rp  = \le(\frac {\epsilon}{4i \mathcal H}\ri)^{\mp \frac \beta \pi} ( 1 + \mathcal O(\epsilon))
\end{align}
This formula shows at once that for $\beta \in (-\pi,\pi)$ the limit of $M_Q(z;\epsilon)$ is the identity matrix and uniformly so in a neighbourhood of $\partial \mathbb D_{\pm 1}$, while for $\beta=\pi$ the $(2,1)$ entry in the conjugation  \eqref{conju} has a limiting value.
The case when $z\in \partial \mathbb{D}_{-1}$ is completely analogous.
The statement about the convergence follows at once by observing that the uniform convergence on compact sets (such as the boundaries of $\mathbb D_{\pm 1}$) implies convergence in all the respective $L^p$ spaces.
\end{proof}

\subsection{Limit of the two-phase solutions}

We are now ready to compute the limit of the two-phase solutions of fNLS $q\lp x,t;\epsilon \rp$ as $\epsilon$ tends to $0^+$. Going backwards through definitions (\ref{81},\ref{Def_Phi_tilde},\ref{Def_Phi} and \ref{Formula_Beni}) one obtains 
\begin{align}\label{Gianfreda}
 q\lp x,t;\epsilon \rp = -2e^{-2 i  d_{\infty}\lp \epsilon \rp}\lim_{z\rightarrow \infty }\ls z\cdot Q\lp z;\epsilon \rp\cdot \psi_0\lp z \rp \rs_{12}.
\end{align}
Let us first assume that $\beta$ belongs to $\lp -\pi;\pi \rp$. The jumps of $Q\lp z;\epsilon \rp$ are defined on a compact contour and in this case they tend to the identity uniformly as $\epsilon$ tends to $0^{+}$. By standard arguments in the theory of Riemann-Hilbert problems \cite{DeiOrt}, then, one has 
\begin{align}
\le\|Q(z;\epsilon)-\1\ri\|  = \frac {o(1)}{1+|z|}, \quad \mbox{uniformly on closed subsets of } \C \setminus \mathbb D_{+1} \cup \mathbb D_{-1}\mbox{ as }\epsilon\rightarrow 0^+.
\end{align}
Formula (\ref{Gianfreda}) reduces then to 
\begin{align}
\lim_{\epsilon\rightarrow 0^{+}} q\lp x,t;\epsilon \rp = -2e^{-2 i  \lim_{\epsilon\rightarrow 0^+} d_{\infty}\lp \epsilon \rp}\lim_{z\rightarrow 	\infty}\ls z\psi_0\lp z \rp \rs_{12}
\end{align}
and theorem \ref{th_onda_piana} follows from definition (\ref{Def_psi_zero}) and proposition \ref{Jenny}.\vspace{1em}\\
The situation is clearly completely different for  $\beta=\pi$. The contribution of $Q\lp z;\epsilon \rp$ to the limit of (\ref{Gianfreda}) as $\epsilon$ tends to zero is in this case not negligible and therefore the matrix $\psi_0$ is {\bf not} the limit of $\bPsi(z;\epsilon)$.  
The strategy is that of solving the ``residual'' RHP for $Q(z;\epsilon)|_{\epsilon=0}$


\begin{RHP}
\label{Eresidual}
Let $\beta =\pi$; the matrix-valued function $Q\lp z;0 \rp$ is the only one satisfying the following properties: 
\begin{description}
\item[\textbf{i}]  $Q\lp z;0 \rp$ is analytic in $\mathbb{C}\backslash\lb \partial \mathbb{D}_1 \cup \partial \mathbb{D}_{-1} \rb.$

\item[\textbf{ii}]  For every $z$ belonging to $\partial \mathbb{D}_1 \cup \partial \mathbb{D}_{-1}$ one has
\begin{align}
Q_{+}\lp z;0 \rp = Q_{-}\lp z;0 \rp M_{Q}\lp z;0 \rp
\end{align}
Here 
\begin{align}
M_Q\lp z;0 \rp = \mathbf 1 + \frac{G\lp z \rp}{z-E_{\pm 3}}, \quad \quad z\in \partial \mathbb{D}_{\pm 1}
\end{align}
where
\begin{align}
G\lp z \rp =\lb 
\begin{array}{ll}
\mathcal{H}e^{-i\ls 2\theta\lp z; x,t \rp + \alpha \rs}\psi_0\lp z \rp\sigma_+\psi_0\lp z \rp^{-1} & \quad\quad z\in\partial \mathbb{D}_1\\
-\overline{\mathcal{H}}e^{i\ls 2\theta\lp z; x,t \rp + \alpha \rs}\psi_0\lp z \rp\sigma_-\psi_0\lp z \rp^{-1} & \quad\quad z\in \partial\mathbb{D}_{-1}
\end{array}
\right.
\end{align}


\item[\textbf{iii}] As $z$ tends to $\infty$ one has
\begin{align}
Q\lp z;0 \rp = \mathbf 1 + \mathcal{O}\lp \frac{1}{z} \rp
\end{align}
\end{description}
\end{RHP}
The solution is constructed explicitly; first of all we are only interested in the explicit expression of $Q_-(z;0)$ (i.e. in the region $\mathbb D_0$) because we need to extract the coefficient of the term $z^{-1}$ at infinity. The second observation is that $Q_-(z;0)$ extends to a meromorphic function with {\em simple poles} at $E_{\pm 3}$; this is seen because $Q_- = Q_+ {M_Q}^{-1}$ and the right side  has a simple pole in $\mathbb D_{\pm 1}$. Therefore we can write
\be
\label{117}
Q_-(z;0) = \mathbf 1 + \frac {A_1}{z-E_3}  + \frac {A_{-1}} {z-E_{-3}}
\ee
The conditions that determine the matrices $A_{\pm 1}$ are the fact that $Q_+ = Q_- M_Q$ must be analytic at $E_{\pm 3}$. 
For example consider the disk $\mathbb D_1$: expanding the product $Q_- M_Q$ in  Laurent series we have 
\begin{align}\label{134}
\lp \mathbf 1 + \frac{A_1}{z-E_3} + \frac{A_{-1}}{z-E_{-3}} \rp\cdot \lp \mathbf 1 + \frac{G\lp z \rp}{z-E_3} \rp =  &\,\,A_{1} \cdot G\lp E_3 \rp\cdot\frac{1}{\lp z-E_3 \rp^2}  + \\ 
& \ls \lp\mathbf 1 + \frac{A_{-1}}{E_3 - E_{-3}}\rp \cdot G\lp E_3 \rp +
 A_{1}\cdot\lp \mathbf 1 + G^{\prime} \lp E_3 \rp \rp  \rs\cdot \frac{1}{z-E_3} +\mathcal{O}\lp 1 \rp,
\end{align}
as $z$ approaches $E_3$. The analyticity condition for $Q_{-}M_Q$ in $E_3$ is equivalent to the following system of matrix equations:
\begin{align} \label{primo_sistema}
\lb\begin{array}{l}
A_1 \cdot G\lp E_3 \rp = 0\\
A_1 \cdot\ls \mathbf 1 + G^{\prime}\lp E_3 \rp \rs + A_{-1}\cdot\frac{G\lp E_3 \rp}{E_3-E_{-3}} + G\lp E_3 \rp = 0
\end{array}\right.
\end{align}
Analyticity condition in $E_{-3}$ yields instead

\begin{align}\label{secondo_sistema}
\lb\begin{array}{l}
A_{-1}\cdot G\lp E_{-3} \rp = 0\\
A_{1}\cdot\frac{G\lp E_{-3} \rp}{E_{-3} - E_3} + A_{-1}\cdot\ls \mathbf 1 + G^{\prime}\lp  E_{-3} \rp \rs + G\lp E_{-3} \rp = 0
\end{array}\right.
\end{align} 
Now, the first equations of (\ref{primo_sistema}) and (\ref{secondo_sistema}) yield respectively
\begin{align}
A_1 = \ls \begin{array}{cc} 0 & a \\ 0 & b \end{array}\rs\psi_0\lp E_3 \rp^{-1}, \quad\quad A_{-1} = \ls \begin{array}{cc} c & 0\\ d & 0 \end{array}\rs\psi_0\lp E_{-3} \rp^{-1}
\end{align}
The complex numbers $a,b,c$ and $d$ are determined using the second  equations of (\ref{primo_sistema}) and (\ref{secondo_sistema}). After appropriate algebraic manipulations one obtains
\begin{align}
c = \overline{b}, \quad\quad d = -\overline{a}, \quad\quad 
a = \frac{V\overline{T} - \overline{W}U}{\la T \ra^2 + \la U \ra^2}, \quad\quad 
b = \frac{\overline{V}U + W\overline{T}}{\la T \ra^2 + \la U \ra^2}.
\end{align}
where
\begin{align}
& T = 1 + \small{\frac{\lp \la E_3-E_1 \ra - \la E_3 - E_{-1} \ra \rp^2}{4\mbox{Im}\lp E_1 \rp\mbox{Im}\lp E_3 \rp}}\,e^{i\ls 2d_0\lp E_3 \rp - 2\theta\lp E_3;x,t \rp -\alpha \rs},\\
& U =   \small{\frac{\mathcal{H}\,\lp\la E_3 - E_1 \ra + \la E_3 - E_{-1} \ra\rp}{4i\,\Imag\lp E_3 \rp \sqrt{\la E_3-E_1 \ra} \sqrt{\la E_3 -E_{-1} \ra}}}\,
e^{i\ls d\lp E_3 \rp - d\lp E_{-3} \rp -2\theta\lp E_3;x,t \rp -\alpha \rs},\\
& V =   -\small{\frac{\mathcal{H}}{2} \lp \sqrt[4]{\frac{E_3-E_1}{E_3-E_{-1}}} + \sqrt[4]{\frac{E_3 - E_{-1}}{E_3 - E_1}} \rp}\, e^{i\ls d_0\lp E_3 \rp - d_{0,\infty} - 2\theta\lp z;x,t \rp - \alpha \rs},\\
& W =   \small{\frac{i\mathcal{H}}{2} \lp \sqrt[4]{\frac{E_3-E_1}{E_3-E_{-1}}} - \sqrt[4]{\frac{E_3-E_{-1}}{E_3 - E_1}}\rp}\, e^{i\ls d_0\lp E_3 \rp + d_{0,\infty} -2\theta\lp z;x,t \rp -\alpha \rs}.\\
\end{align}

\begin{corollary}
\label{cor28}
Let $\beta$ equal $\pi$ and define the matrix 
\begin{align}\label{Eneide}
\mathcal E(z;\epsilon):= Q(z;\epsilon) Q^{-1}(z;0).
\end{align}
Then
\be
\le\|
\mathcal E(z;\epsilon) -\1
\ri\|  = \frac {o(\epsilon \ln \epsilon )}{1+|z|}
\ee
uniformly over closed sets of $\mathbb C\setminus  \lb\mathbb D_{1} \cup \mathbb D_{-1}\rb$.
\end{corollary}
{\bf Proof.}
We have proved that (for $\beta=\pi$) 
we have 
\be
\le\|M_Q(z;\epsilon) M^{-1}_Q(z;0) - \1\ri\|  = \mathcal O(\epsilon \ln \epsilon)\ ,\qquad 
z\in \partial \mathbb D_{\pm 1} 
\ee
where the convergence is in $\sup$ norm and hence also (and at the same rate) in any $L^p$. 
The matrix $\mathcal E\lp z;\epsilon \rp$ is normalized to the identity at infinity, and its jumps, only on $\partial\mathbb D_{\pm 1}$, are of the form
\be
\mathcal E_+(z;\epsilon) = \mathcal E_-(z;\epsilon) \ls Q_{-} (z;0)M_Q(z;\epsilon) M^{-1}_Q(z;0) Q^{-1}_{-}(z;0)\rs = 
\mathcal E_-(z;\epsilon)  \le(\1 + { \mathcal O(\epsilon \ln \epsilon) }\ri)
\ee
This is sufficient to apply the standard small norm theorem for Riemann-Hilbert problems.
$\Box$

\vspace{1em}
Substituting after definition (\ref{Eneide}) in (\ref{Gianfreda}), and taking the limit on both sides for $\epsilon$ that goes to zero yields
\begin{align}\label{Leone}
\lim_{\epsilon\to 0^+}q\lp x,t;\epsilon \rp = -2e^{-2i\lim_{\epsilon\to 0^+}d_{\infty}\lp \epsilon \rp}\cdot \lim_{\epsilon\to 0^+}\lim_{z\to \infty}\ls z\cdot\mathcal{E}\lp z;\epsilon \rp\cdot Q\lp z;0 \rp\cdot \psi_0\lp z \rp \rs_{12}
\end{align}
The uniform convergence guaranteed by corollary \ref{cor28} allows one to exchange the two limits in the r.h.s. above: 
\begin{align}
\lim_{\epsilon\to 0^+}\lim_{z\to \infty}\ls z\cdot \mathcal{E}\lp z;\epsilon \rp\cdot Q\lp z;0 \rp\cdot\psi_0\lp z \rp \rs_{12} &= \lim_{z\to\infty}\lim_{\epsilon\to 0^+}\ls z\cdot \mathcal{E}\lp z;\epsilon \rp\cdot Q\lp z;0 \rp\cdot\psi_0\lp z \rp \rs_{12}\\
&= \lim_{z\to \infty}\ls z\cdot Q\lp z;0 \rp\cdot \psi_0\lp z \rp \rs_{12}.
\end{align}
In view of this last one and using \eqref{117}, then  (\ref{Leone}) simplifies to
\begin{align}
\lim_{\epsilon\to 0^+}q\lp x,t;\epsilon \rp = -2e^{-2i\lim_{\epsilon\to 0^+}d_\infty\lp \epsilon \rp}\lb \ls A_1 + A_{-1}\rs_{12} + \lim_{z\to\infty}\ls z\psi_0\lp z \rp \rs_{12} \rb.
\end{align}
The proof of theorem \ref{th_Soliton} is then completed by means of straightforward though long algebraic manipulations.
\vspace{2em}

\end{document}